\documentclass[letterpaper]{article}
\usepackage[preprint]{aaai2027} 
\usepackage[hyphens]{url} 
\usepackage{graphicx} 
\usepackage{natbib} 

\newcommand*{\textcite}{\citet}

\usepackage{caption} 
\usepackage{amsmath}
\usepackage{amssymb}
\usepackage{amsthm}
\usepackage{blkarray}
\usepackage{mathtools}

\usepackage{thmtools} 
\usepackage[capitalise]{cleveref}

\theoremstyle{plain} 
\newtheorem{theorem}{Theorem} 
\newtheorem{proposition}[theorem]{Proposition}

\theoremstyle{definition} 
\newtheorem{definition}[theorem]{Definition}
\newtheorem{example}[theorem]{Example}

\usepackage{todonotes}
\usepackage{xstring}

\newcommand{\pardo}[1]{}

\title{Vertex cover number of valued constraints is a structural parameter \\ for efficient local search}
\author{Artem Kaznatcheev}
\affiliations{
Department of Mathematics, and \\
Department of Information and Computing Sciences,\\
Utrecht University
}

\begin{document}

\maketitle

\begin{abstract}
Many local search methods for problems in artificial intelligence can be viewed as an uphill climb on a corresponding discrete fitness landscapes.
Finding even local peaks in these fitness landscapes is computationally intractable in theory, but often works in practice.
So what features of fitness landscapes allow for efficient local search?
Since any fitness landscapes can be represented by a (hyper)graph of valued constraints, 
I re-frame this as a question of parameterized complexity:  
what structural parameter of a valued constraint graph guarantees that a strict local search will find a local peak in the corresponding fitness landscape efficiently?
Given a valued constrain graph of vertex cover number $k$, I prove that greedy local search will find the a local fitness peak in at most $2^{2k}\cdot(n - k + 1)$ steps and random uphill local search will find a local fitness peak in an expected number of at most $2^k\cdot n(n - k)$ steps.
I also show that these results are asymptotically good for strict local search because there are valued constraint graphs of vertex cover number $k$ where every ascent from some initial assignment has a length of $\frac{9}{128} \cdot 2^k \cdot (n - k)$ or greater.
This suggests vertex cover number as a good structural parameter for the complexity of local search.
\end{abstract}

\section{Introduction}

\pardo{Fitness landscapes}

Many problems in artificial intelligence can be re-framed as search or optimization in an appropriately defined combinatorial search space~\cite{AIbook,AISearch2023}.
Such combinatorial optimization problems are often defined over some set $V$ of Boolean variables with the goal of finding an assignment $x^* \in \{0,1\}^V$ that maximizes some pseudo-Boolean objective function $f: \{0,1\}^V \rightarrow \mathbb{Z}$~\cite{pseudoBool,QPB}.
I will refer to these $f$ as \emph{fitness functions}.
In many cases, these fitness functions have a structure such that `similar' assignments tend to yield similar fitness.
This makes it natural to introduce an adjacency structure based on just assignments themselves, independent of the particular fitness function.
The simplest such structure is to call two assignments \emph{adjacent} if they differ at a single bit.
Or more precisely: I will say that $x$ and $y$ are adjacent if there $\exists u \in V$ such that $x_u \neq y_u$ and $\forall v \in V \backslash \{u\}, \; x_v = y_v$.
An assignment $x^*$ is a \emph{local peak} if for all adjacent assignments $y$, $f(x) \geq f(y)$, and a \emph{global} peak if this inequality is satisfied for all assignments $y \in \{0,1\}^V$.
A \emph{fitness landscape} is a pseudo-Boolean function together with this assignment adjacency structure~\cite{W32,evoPLS,KazThesis,repCP}.

It is well known that finding global peaks in fitness landscapes is NP-hard, but even finding local peaks is difficult.
\textcite{PLS} introduced the complexity class of polynomial-local search (PLS) to encode the problem of finding \emph{any} local peak, not just the global one.
Fitness landscapes are PLS-complete even if we restrict to $f$ that are at most quadratic~\cite{W2SAT_PLS1,W2SAT_PLS2,PLS_VCSP2013,PLS_Survey}.
Thus, unless FP = PLS (which is something that most researchers do not believe to be likely), finding \emph{any} local peak in a general fitness landscapes is computationally intractable for any polynomial-time algorithm.

\pardo{From local search to strict local search}

Since we cannot hope for a polynomial time algorithm to solve arbitrary combinatorial optimization problems,
researchers and practitioners in AI often turn to local search heuristics~\cite{LocalSearch_Book1,LocalSearch_Book2,AIbook}.
Local search starts with some initial assignment $x^0 \in \{0,1\}^V$ and follows some update rule to proceed to another assignment until it eventually finds a good one.
I will call such a search a \emph{strict local search} if each step of the algorithm produces an adjacent assignment of higher fitness.
The sequence of such assignments $x^0,x^1,\ldots,x^T$ is an \emph{ascent}: each $x^{t + 1}$ is adjacent to $x^t$, $f(x^{t + 1}) > f(x^t)$, and $x^T$ is a local peak.

\emph{Greedy} and \emph{random uphill} are two particularly popular local search algorithms~\cite{LocalSearch_Book1,approxBook,AIbook} that follow \emph{steepest} and \emph{random} ascents, respectively, in the fitness landscape~\cite{pw2MSc,effective-and-efficient}.
Given a current assignment $x^t$, let $X^t = \{x' \; | \; f(x') > f(x^t) \text{ and } x' \text{ is adjacent to } x\}$ be the set of adjacent assignments of higher fitness.
Greedy local search algorithm selects an adjacent assignment $x^{t + 1} = \text{argmax}_{x' \in X} f(x')$ of $x^t$ such that all other assignments $x'$ adjacent to $x^t$  have $f(x') \leq f(x^{t + 1})$.
I call an ascent followed by this algorithm a \emph{steepest ascent}.
A lot of work has been done on constructing simple fitness landscapes with provably long steepest ascents~\cite{fittestHard2,hakenSteepest,fittestHard3,tw7,pw2MSc,slow-greed}.
Random uphill local search can be seen as a stochastic relaxation of greedy that selects an adjacent $x^{t + 1}$ uniformly at random from $X^t$ -- this local search is also known as Random-Edge in the simplex literature~\cite{randomFitter1,randomFitter2,effective-and-efficient}.
So if I want to show that greedy and random uphill local search are \emph{efficient} at finding a local peak, I need to upper bound the length of the steepest ascent or the expected length of the random ascents.
And if I want to show that both of these local search algorithms, as well as all other strict local search algorithms, are \emph{inefficient} then I can do this by lower bounding the length of \emph{all ascents} from some initial assignment.

In this paper, I will use parameterized complexity to study when greedy and random uphill local search are efficient.
I will do this in three steps.
First, I will review how to represent fitness landscapes by instances of valued constraint satisfaction problems (VCSPs)~\cite{repCP}.
Specifically, I will focus on fitness landscapes represented by VCSP on $n$ variables with a constraint (hyper)graph that has vertex cover number $k$.
Second, I will provide bounds on steepest and random ascent, upper bounding the former by $\leq 2^{2k}\cdot(n - k + 1)$ (\cref{thm:steepest_ascent}) and the latter by $\leq 2^k\cdot n(n - k)$ (\cref{thm:random_ascent}).
This will allow me to prove that greedy local search solves VCSPs in fixed-parameter linear time and random uphill local search solves VCSPs in fixed-parameter quadratic time when parameterized by vertex cover number.
As far as I know, these are the first parameterized complexity tractability results for strict local search algorithms.
Third, I will show that these bounds are not loose by showing how to construct a VCSP where all ascents (thus including steepest and random) have length $\geq \frac{9}{128}\cdot2^k\cdot(n - k)$ from some initial assignment (\cref{ex:shortest_ascent}).

\section{Representing fitness landscapes \\ by valued constraint hypergraphs}

\pardo{Definition of pseudo-Boolean functions to VCSPs}
Any pseudo-Boolean function $\mathcal{C}: \{0,1\}^V \rightarrow \mathbb{Z}$ can be uniquely written down as a polynomial:
\begin{equation}
\mathcal{C}(x) = \sum_{S \in \mathcal{S}}w(S)\prod_{u\in S}x_u
\label{eq:gen_poly}
\end{equation}
where $\mathcal{S} \subseteq 2^V$ is a \emph{set of scopes},
$w: \mathcal{S} \rightarrow \mathbb{Z} \setminus \{0\}$ is the \emph{constraint weight function},
and each \emph{scope} $S \in \mathcal{S}$ is a set of variables $x_u$ with $u \in S$ that are involved in an AND-constraint of weight $w(S)$.
In this representation, (locally) maximizing $f$ is known as a Boolean \emph{valued constraint satisfaction problem (VCSP)}
with a specific $\mathcal{C} = (V,\mathcal{S},w)$ as the VCSP-instance.
Given a specific constraint with scope $S \in \mathcal{S}$, since we are focused on maximization, if $w(S) > 0$ then I will say $S$ is \emph{satisfied} when $\forall u \in S, x_u = 1$ and if $w(S) < 0$ then I will say $S$ is \emph{satisfied} when at least one $u \in S$ has $x_u = 0$ -- otherwise I will say the constraint is \emph{unsatisfied}.
From this perspective, every fitness landscape has a unique corresponding VCSP-instance that represents it~\cite{repCP}.

\pardo{Move to constraint (hyper)graphs: Discussion of prior parameters like treewidth being insufficient and connection to non-local algorithms}
To get insight into the structural complexity of fitness landscapes via the VCSP-instances that represent them, I will also read $\mathcal{C} = (V,\mathcal{S},w)$ as a weighted (hyper)graph where $V$ are the vertexes, $\mathcal{S}$ is the set of hyperedges, and $w$ is the weight function on hyperedges.
Many NP-hard search problems become tractable by dynamic programming on tree decompositions when the problem instances have small treewidth~\cite{treewidth1988,treewidth2007,treewidth2008,treewidth2011,hypertree2016}.
This is also true for valued constraint satisfaction -- VCSPs are fixed-parameter tractable when parameterized by treewidth~\cite{BB73,VCSPsurvey}.
Unfortunately, this does not carry over from dynamic programming to strict local search.
There are valued constraint graphs of bounded treewidth (in fact, of pathwidth $\leq 3$) where all ascents from some initial assignment are exponentially long~\cite{MaxCutDeg4,MaxCutSimple,VCSPpw3_allExp}.
Thus, treewidth is too coarse grained a structural parameter for efficient local search.
Here, I show that vertex cover number is a better parameter.

\pardo{Definition of induced subproblem}
For a more natural definition of vertex cover number for VCSP and the fitness landscapes they represent, it is helpful to first introduce the idea of an induced subproblem from \textcite{star_VCSPs} and \textcite{VCSPpw3_allExp}:

\begin{definition}[Induced subproblem]
Given a VCSP-instance $\mathcal{C} = (V,\mathcal{S},w)$, any subset $V' \subset V$, and background assignment $z \in \{0,1\}^{V \setminus V'}$,
let $\mathcal{S}' = \{S \cap V' | S \in \mathcal{S}\}$ and $W': \mathcal{S}' \rightarrow \mathbb{Z}$ given by 
\begin{equation}
w'(R)= \sum_{Q \in \mathcal{N}(R,V\setminus V')} w(R \cup Q) \prod_{i \in Q} y_i
\label{eq:subW}
\end{equation}
where $\mathcal{N}(R,V\setminus V') = \{Q \; | \; Q \in V \setminus V' \text{ and } R \cup Q \in \mathcal{S}\}$ is the set of (hyper-edge completion) neighbors of $R$ in $V\setminus V'$,
then $\mathcal{C}' = (V',\mathcal{S'},W')$ is the \emph{induced subproblem} of $\mathcal{C}$ on $V'$ in background $z$.
\end{definition}
\noindent We can think of a subproblem on $V'$ in background $z$ as restricting the scope of all constraints down to $V'$ by fixing the partial assignment to $z$ on the variables being removed.
From the perspective of the fitness landscape implemented by the VCSP, a VCSP-subproblem select a sublandscape or `face' of the fitness landscape by fixing the dimensions in $V \setminus V'$ to the partial assignment $z$:
\begin{proposition}[\textcite{VCSPpw3_allExp}]
If $\mathcal{C}' = (V',\mathcal{S'},W')$ is an induced subproblem of $\mathcal{C} = (V,\mathcal{S},W)$ on $V'$ in background $z$, we have $\mathcal{C'}(y) = \mathcal{C}(yz)$ for all $y \in \{0,1\}^{V'}$.
\end{proposition}
From the perspective of the hypergraph of constraints, an induced subproblem on $V'$ eliminates the vertexes in $V \setminus V'$, any hyperdge involving only those vertexes, and restricting the hyperedges that cross the cut from $V'$  to $V \setminus V'$ down to just $V'$.
If the result leaves only unary (and nullary) constraints (i.e., at most `self-loops' in the hypergraph) then this captures the classic idea of $V \setminus V'$ being a vertex cover

\pardo{Definition of vertex cover number}
\begin{definition}(Vertex cover number; \textcite{star_VCSPs})
A VCSP $\mathcal{C} = (V,\mathcal{S},w)$ has \emph{vertex cover number} $k$ (writen as $\text{vertexcover\#}(\mathcal{C}) = k$) if there exist $k$ variables $K \subseteq V$ such that the induced subproblem $\mathcal{C}^z$ of $\mathcal{C}$ on $V \setminus K$ in any background $z \in \{0,1\}^K$ has at most unary constraints.
The set of variables $K$ is known as the \emph{vertex cover}.
\end{definition}

Vertex cover number is a more fine-grained parameter than treewidth and pathwidth: 
$\text{treewidth}(\mathcal{C}) \leq \text{pathwidth}(\mathcal{C}) \leq \text{vertexcover\#}(\mathcal{C})$. 
A star-structured VCSP has vertex cover number $1$.
We can think of a VCSP with vertex cover number $k$ as a starlike graph that allows the centre to consist of $k$ variables while still restricting the other $n - k$ `leaf' variables to form an independent set.

\pardo{Specifying VCSPs with large independent set in terms of induced subproblems}
For a VCSP-instance $\mathcal{C} = (K \cup M,\mathcal{S},w)$ with a vertex cover $K$, we know that it implements a pseudo-Boolean function of degree at most $|K| + 1$ and has $|\mathcal{S}| \leq 2^{|K|}\cdot(|M|+1)$.
But instead of specifying the weight of all these constraints for the high degree polynomial in \cref{eq:gen_poly}, we can equivalently specify it as a collection of $2^{|K|}$ linear functions $\mathcal{C}^z$ for each $z \in \{0,1\}^K$:
\begin{equation}
\mathcal{C}^z(y) = w^z(\emptyset) + \sum_{u \in M}w^z(\{u\})
\end{equation}
which come together as the general pseudo-Boolean function:
\begin{equation}
\mathcal{C}(x) = \sum_{z \in \{0,1\}^K} \delta(x[K] = z) \cdot \mathcal{C}^z(x[V \setminus K])
\end{equation}
where the delta function $\delta(x[K] = z)$ is the polynomial:
\begin{equation}
\prod_{v \in K}(z_vx_v + (1 - z_v)(1 - x_v)).
\end{equation}
Thus, a VCSP-instance $\mathcal{C} = (K \cup M, \mathcal{S}, w)$ with vertex cover $K$ can be uniquely specified by the upto $2^{|K|}(|M| + 1)$ weights $\{w^z(\emptyset) \; | \; z \in \{0,1\}^V\} \cup \{w^z(u) \; | \; u \in M \text{ and } z \in \{0,1\}^K\}$.
I will use this representation of fitness landscapes to prove that steepest ascents have length of at most $2^{2|K|}\cdot(n - |K| + 1)$ (\cref{thm:steepest_ascent}), 
random ascents have an expected length of at most $2^{|K|}\cdot n(n - |K|)$ (\cref{thm:random_ascent}), 
and to build a VCSP-instance where all ascents from $0^{|K|}0^{|M|}$ have length of at least $\frac{9}{128}\cdot 2^{|K|} \cdot (n - |K|)$ (\cref{ex:shortest_ascent}).

\section{Upper bound on steepest and random ascent}

\textcite{star_VCSPs} introduced the vertex cover number parameter for VCSPs and showed that from any initial assignment in any VCSP-instance on $n$ variables with vertex cover number $k$, there exists some ascent of length $\leq 2^k(n - k + 1) - 1$.
But their existence proof does not provide a specific local search algorithm that could actually follow this ascent (at least not one that does not have a small vertex cover in hand).
Here, I will show that greedy local search and random uphill local search also find short ascents
while treating the fitness function as a black-box 
(so without any knowledge of the vertex cover).
I will only use knowledge of the vertex cover for the analysis of the runtimes, not in the running of the algorithm itself.
I will study the behavior of ascents on the fitness landscape $\mathcal{C}$ on $n$ variables with vertex cover number $k$ by decomposing the landscapes into the $2^k$ sublandscapes induced by each partial assignment $z \in \{0,1\}^K$ to the vertex cover $K$ and seeing how the ascents move through each sublandscape.
Specifically, I will show that steepest ascent makes at most $n - k$ steps in each induced sublandscape and random ascent makes an expected number of $\leq n(n - k)$ steps in each induced sublandscape.
This will show that steepest ascent solves VCSPs in fixed-parameter linear time and random ascent solves VCSPs in fixed-parameter quadratic time when parameterized by vertex cover number.

\begin{theorem}
If a Boolean VCSP-instance $\mathcal{C}$ has vertex cover number $k$ then the steepest ascent is $\leq 2^{2k}\cdot(n - k + 1)$
\label{thm:steepest_ascent}
\end{theorem}

\begin{proof}
Let $K \subseteq V$ be a vertex cover of size $k$ and let $m = n - k$ be the number of other variables.
For each $z \in \{0,1\}^K$, define $\mathcal{C}^z$ as the induced subproblem on $V \setminus K$ with background $z$.
Since $K$ is a vertex cover, it means there are only unary and nullary constraints in each $\mathcal{C}^z$, number the unary constraints for each $\mathcal{C}^z$ based on their magnitude $|w^z_1| \geq |w^z_2| \geq \ldots \geq |w^z_m|$~\footnote{
note $w^z_1$ and $w^{z'}_1$ might act on different variables, that is why I don't use the $w^z(\{1\})$ notation here.
Formally: for each $z$ there is a permutation $\sigma^z$ such that $w^z_i = w^z(\{\sigma^z(i)\})$.
} and define the partial sums $w^z_{\leq j} = \sum_{i = 1}^j |w^z_i|$
From this, we can define the lowest fitness in $\mathcal{C}^z$ -- for an assignment of $1$ for each variable with negative weight, and $0$ for positive weighted variables -- as:
\begin{equation}
\text{min}^z = w^z(\emptyset) + \sum_{\substack{i \in [m] \\ \text{ s.t. } w^z_i < 0}} w^z_i
\end{equation}
and highest fitness assignments (assign $1$ for each variable with positive weight, and $0$ for negative weighted variables) in $\mathcal{C}^z$ as $\text{min}^z + w^z_{\leq m}$.
Now we can show that:

\textbf{Claim:} Along any steepest ascent where $j$ steps have happened in $V \setminus K$ with background $z$ if the current assignment is $x = yz$ then $\mathcal{C}^z(y) \geq \text{min}^z + w^z_{\leq j}$. 

We prove this claim by induction.
It is true for $j = 0$ because $\text{min}^z$ is the minimum fitness in this sublandscape.
Now consider the case of $j + 1$.
Let $y'$ be the assignment right before the $j + 1$th flip in $\mathcal{C}^z$ and $y$ be the assignment after.
If $\mathcal{C}^z(y') \geq \text{min}^z + w^z_{\leq j + 1}$ then we are done.
Otherwise, $w^z_{\leq j + 1} > \mathcal{C}^z(y') - \text{min}^z \geq  w^z_{\leq j}$ where the second inequality is from the inductive hypothesis.
From this, it follows that at least one of the $j + 1$ largest constraints by magnitude is unsatisfied.
Since, steepest ascent flips the largest available increase, it will satisfy the largest unsatisfied constraint and so increase fitness by $\geq |w_{j + 1}|$ thus resulting in $\mathcal{C}^z(y) \geq |w_{j + 1}| + w_{\leq j} + \text{min}^z = \text{min}^z + w_{\leq j + 1}$.

This claim implies that there are at most $m$ steps in each $\mathcal{C}^z$.
This means there are at most $2^k\cdot m$ many flips of variables in $M$, so the set $Y \subseteq \{0,1\}^M$ that actually occur in the ascent must have size $|Y| \leq 2^k\cdot m + 1$.
For a given $y \in Y \subseteq \{0,1\}^M$, each $z \in \{0,1\}^K$ can be visited at most once without violating the increasing fitness property of an ascent.
This gives us an overall bound on the length of a steepest ascent of:
\begin{align}
\leq \underbrace{|Y| - 1}_{\text{flips in } M} + \underbrace{|Y|(2^k - 1)}_{\text{flips in } K} & = 2^k(2^k\cdot m + 1) - 1 \\
& \leq 2^{2k}(n - k + 1)
\end{align}
\end{proof}

A similar fixed parameter tractability theorem with similar proof holds for random uphill local search:

\begin{theorem}
If a Boolean VCSP-instance $\mathcal{C}$ has vertex cover number $k$ then a random ascent has expected length $\leq 2^{k}\cdot n(n - k)$.
\label{thm:random_ascent}
\end{theorem}

\begin{proof}
As with the proof of \cref{thm:steepest_ascent}, define $K$, $m$, $\mathcal{C}^z$, $w^z_i$, $w^z_{\leq j}$, and $\text{min}^z$.
Suppose our random ascent has just flipped a bit to reach a $yz$ with $w^z_{\leq j + 1} > \mathcal{C}^z(y) - \text{min}^z \geq w^z_{\leq j}$.
At this point, there is at least one unsatisfied constraint in $\mathcal{C}^z$ with weight $\geq w_{j + 1}$, as with the steepest ascent case, flipping the variable associated with such an unsatisfied constraint will take the fitness in $\mathcal{C}^z$ above $w^z_{\leq j + 1} + \text{min}^z$.
Since there are at most $n - 1$ other variables that might be flippable while increase fitness, this means we expect to flip $\leq n$ to get $\mathcal{C}^z$ to increase to the next level (this might end up taking several returns to $\mathcal{C}^z$ when variables in $K$ are flipped).
Since there are $m = n - k$ levels of $w^z_{\leq j}$s, this means an expected total number of flips of $n(n - k)$ per background $z$.
Combine this with with the number of assignments for $z \in \{0,1\}^K$ to get an upper bound on the expected number of steps of $\leq 2^k\cdot n (n - k)$.
\end{proof}

Of course, a more careful analysis of the arguments presented in the above two proofs could be undertaken to reduce some of the upper bounds slightly, 
but this will not lead to a drastic qualitative reduction of the bounds.

\section{Lower bound on shortest ascent}

To show that the bounds in \cref{thm:steepest_ascent,thm:random_ascent} are not loose, I will give an example of a VCSP on $n$ variables of vertex cover number $k$ where all ascents from $0^n$ have length at least $\frac{9}{128}\cdot2^k\cdot (n - k)$ (\cref{ex:shortest_ascent}).
This will also establish that \textcite{star_VCSPs}'s prior upper bound of $\leq 2^k\cdot(n - k + 1) - 1$ on the shortest ascent parameterized by vertex cover number is tight.
My main trick in \cref{ex:shortest_ascent} is to use the variables in the vertex cover to place an existing construction of all ascents long (of length $2^{\Omega(k)}$) but adjust the fitness function in such a way that every second step of that existing long ascent is not uphill unless all the variables in the independent set are flipped.
I could place any existing construction of all ascents long (for example, I could use any one of \textcite{MaxCutDeg4}, \textcite{MaxCutSimple}, or \textcite{VCSPpw3_allExp}) in the vertex cover, but I can get the simplest and tightest result by using a snake-in-the-box code~\cite{snakeLowerBound}.
\pardo{Intro for how to lower bound steepest and random ascent by making all ascents long}

\pardo{Main trick is to put an existing all ascents long construction at the centre and have it wiggle the leaves.
Existing constructions tend to have just one ascent that any strict local search algorithm follows.}

\pardo{A single ascent can be abstracted into a snake in the box codes.}
A \emph{snake} on $k$ bits is a sequence of $T$ assignments $\{x^1,\ldots,x^T\} \subseteq \{0,1\}^k$ such that for all $i,j \in [T]$, $x^i$ is adjacent to $x^j$ if and only if $j = i \pm 1$. 
In other words, each assignment in the snake is adjacent to exactly two other assignments in the snake, 
except for the head ($x^1$) and the tail ($x^T$) that each have only one adjacent assignment in the snake. 
We can always assume that the head $x^1 = 0^k$.
Let $\text{sn}: \{0,1\}^k \rightarrow [T] \cup \{\bot\}$ be a function that maps each assignment to its position in the snake, or to $\bot$ if it is not in the snake:
\begin{equation}
\text{sn}(x) = \begin{cases}
t & \text{if } x = x^t \\
\bot & \text{if } x \not\in \{x^1,\ldots x^T\}
\end{cases}
\end{equation}
\textcite{snakeLowerBound} showed the there exists snakes on $k$-bits with length $T$ lower bounded by $\frac{9}{64}2^k$.

\pardo{Construction of longest ascent}
\begin{example}\label{ex:shortest_ascent}
Consider a VCSP on $n = k + m$ variables $K \cup M$ with $|K| = k$ and $|M| = m$.
Let $\text{sn}: \{0,1\}^K \rightarrow [T] \cup \{\bot\}$ be a snake on $k$-bits with assignments $X = \{x^1,x^2,\ldots,x^T\}$ with $x^1 = 0^k$.
For each $x \in \{0,1\}^K$, define the induced constraint weights $w_{\text{sn}(x)}$:
\begin{align*}
w_{\bot}(\emptyset) \quad\; & = - 1 & \; & w_{\bot}(\{i\}) \quad\; = -1 \\
w_{4j+1}(\emptyset) & = k(2m + 4) & \; & w_{4j+1}(\{i\})  = \;\;\; 1 \\
w_{4j+2}(\emptyset) & = k(2m + 4) + 1 - 2m & \; & w_{4j+2}(\{i\})  = \;\;\; 3 \\
w_{4j+3}(\emptyset) & = k(2m + 4) + 2m + 2 & \; & w_{4j+3}(\{i\})  = -1 \\
w_{4j+4}(\emptyset) & = k(2m + 4) + 2m + 3 & \; & w_{4j+4}(\{i\})  = -3
\end{align*}
where $w_{\text{sn}(x)}(\{i\})$ is the same for each $i \in M$.
The important feature to notice is that for all $y \neq 1^m$: $\mathcal{C}^{4j + 1}(y) > \mathcal{C}^{4j + 1}(y)$ but $\mathcal{C}^{4j + 1}(1^m) < \mathcal{C}^{4j + 1}(1^m)$.
And similarly for all $y \neq 0^m$: $\mathcal{C}^{4j + 3}(y) > \mathcal{C}^{4j + 4}(y)$ but $\mathcal{C}^{4j + 3}(0^m) < \mathcal{C}^{4j + 4}(0^m)$.
Thus, steps in $K$ along the snake between $x^{4j + 1}$ and $x^{4j + 2}$; and between $x^{4j + 3}$ and $x^{4j + 4}$ cannot be made without making all $m$ steps in $M$ to bring $\mathcal{C}^{4j + 1}$ and $\mathcal{C}^{4j + 3}$, respectively, to their peaks.
This `blocking' by the independent set of variables is the main trick.

Now consider an ascent starting from $0^k0^m$.
This ascent will go through all the $x^t$ in the snake $X$ and flip all the bits in $M$ after every two assignments in the snake.
Specifically, the ascent will have the following structure:
\begin{align}
x^{4j + 1}0^m \xrightarrow[]{\;\; m \text{ flips in } \mathcal{C}^{4j + 1} \;}  & \; x^{4j+1}1^m \label{eq:phase1}\\
x^{4j + 1}1^m \rightarrow x^{4j + 2}1^m \rightarrow & \; x^{4j + 3}1^m \\
x^{4j + 3}1^m \xrightarrow[]{\;\; m \text{ flips in }  \mathcal{C}^{4j + 3} \;} & \; x^{4j + 3}0^m \\
x^{4j + 3}0^m \rightarrow x^{4j + 4}0^m \rightarrow & \; x^{4j + 5}0^m \label{eq:phase4}
\end{align}
which leads to a total ascent length of $T\cdot\frac{m}{2}$.
Combined with the \textcite{snakeLowerBound} bounds on $T$ this results in an ascent of length $\geq \frac{9}{128} \cdot 2^k \cdot m$.
Since no other uphill steps are possible starting from $0^n$ outside of those specified by Equations~\eqref{eq:phase1}-\eqref{eq:phase4}, this gives us the lower bound on the shortest ascents from some initial assignment in a VCSP of vertex cover number $k$.
\end{example}

\pardo{Comments on other constructions for longest ascent}

\section{Conclusion}
Together \cref{thm:steepest_ascent}, \cref{thm:random_ascent}, and \cref{ex:shortest_ascent} tell us that if we want to understand what makes fitness landscapes difficult or easy for local search to navigate then we can start by looking at the vertex cover number of the VCSP-instances that represent these fitness landscapes.
Of course, this is only a first step in this emerging field of parametrized complexity of local search.
In general, it would be interesting to see what other structural parameters of combinatorial optimization problems become relevant as we follow the long ascent to understanding this complexity.
And in the particular cases in AI where local search seems to work efficiently in practice, can we show that these structural parameters are bounded?

\section*{Acknowledgments}
I would like to thank Dave Cohen and Peter Jeavons for helpful discussions and feedback on early drafts of this paper.

\bibliography{VCSP_vcn}

\end{document}